\documentclass[12pt]{article}
\usepackage{amsfonts,amsthm, latexsym, amssymb, mathrsfs, graphicx}
\usepackage[small]{titlesec}
\newcommand{\A}{{\rm A}}
\newcommand{\B}{{\rm B}}

\def\O{{\rm O}}
\def\d{{\rm d}}
\newcommand{\R}{{\mathbb R}}

\begingroup
\newtheorem{theorem}{Theorem}[section]
\newtheorem*{theorem*}{Theorem}
\newtheorem{lemma}[theorem]{Lemma}
\newtheorem{proposition}[theorem]{Proposition}
\newtheorem{corollary}[theorem]{Corollary}
\newtheorem{definition}[theorem]{Definition}

\endgroup

\title{\large{\bf Some simple results about the Lambert problem}}

\author{}

\date{}

\begin{document}
	
	\maketitle
	
	\begin{center}
	{\bf Alain Albouy$^{1}$,\qquad Antonio J. Ure\~na$^{2}$}
	
	\bigskip
	$^{1}$ IMCCE, CNRS-UMR8028, Observatoire de Paris
	
        77, avenue Denfert-Rochereau, 75014 Paris, France
         
         Alain.Albouy@obspm.fr
         
         \bigskip
         
         $^{2}$ Departamento de Matematica Aplicada, Facultad de Ciencias
         
         Campus Universitario de Fuentenueva
         
         Universidad de Granada, 18071, Granada, Spain
         
        ajurena@ugr.es
        
        \bigskip
	\end{center}

\section*{Abstract} We give simple proofs of some simple statements concerning the Lambert problem. We first restate and reprove the known existence and uniqueness results for the Keplerian arc. We also prove in some cases that the elapsed time is a convex function of natural parameters. Our statements and proofs do not distinguish between the three types of Keplerian conic section, elliptic, parabolic and hyperbolic. We also prove non-uniqueness results and non-convexity results. We do not develop any algorithm of resolution, limiting ourselves to such obviously useful {\it a priori} questions: How many solutions should we expect? Can we be sure that the Newton method will converge?

\section{Introduction}

The Lambert problem is a boundary value problem for the Kepler problem. It is stated as follows: {\it find the Keplerian arcs around a fixed center $\O$ which go from a given point $\A$ to a given point $\B$ in a given elapsed time $T$.}
The problem was briefly posed by Lambert in 1761 in a letter to Euler (see \cite{alb}, \cite{bopp}). Gauss posed it again in his {\it Theoria motus} \cite{gauss}, \S 84, and proposed two numerical methods to solve it quickly and accurately (\S 85--87 and \S88--105). The problem was named the Lambert problem in the 1960's since it is closely related to Lambert's theorem \cite{lambert}. Lagrange~\cite{lagrange4} did not study the Lambert problem, but he proposed trigonometric formulas to reprove Lambert's theorem, and since Gauss, these formulas are commonly used to solve the Lambert problem.

Many authors studied the problem, first in the same context as Gauss, i.e., as an ingredient in a method of orbit determination, then in other contexts related to the space conquest. This abundant literature is mainly concerned with the development of algorithms. Only few authors tried to improve our {\it a priori} knowledge of the number of solutions, the reason being probably that the first attempts which present themselves are usually blocked by technical difficulties. The usual Keplerian recipe proposes elementary, but complicated expressions, which furthermore change according to the type of conic section. The object of our work is to show the effectiveness of other ideas.

An {\it a priori} knowledge of the number of solutions of the Lambert problem is clearly useful when conceiving an algorithm to find them. But it may also address the question: after separating at a point, may two orbiting bodies meet again at another point?

The number of solutions of the Lambert problem is carefully discussed by Eliasberg in his book \cite{eli}. He gives very complete results, including the number of roots corresponding to arcs with several revolutions. Many of his arguments consist in direct computations and observations on classical series. But apparently he also argues on graphics drawn for particular values of the parameters (see top of his page 116). We do not consider that his results are completely established.

Lancaster and Blanchard \cite{lb} draw a clear figure where the number of roots appear in all cases, as well as the places where the function $T$ is not convex. We will relate their variables to the rectilinear case and confirm their results. 

Sim\'o \cite{simo} strongly influenced our work. He is the only author we know to state, as a theorem, an existence and uniqueness result for the Lambert problem. He insists on convexity results as a guaranty of convergence of the Newton method. We will show that these results are also useful to count rigorously the arcs with several revolutions.

In a recent work, De La Torre, Flores and Fantino \cite{dff} present and continue Sim\'o's work, after a detailed review of the literature on the Lambert problem.

\section{Uniqueness of the direct symmetric arc}

\begin{definition}
In the plane $\O xy$, a Keplerian arc around the center $\O$  is an orbit of the Newton equation
\begin{equation}
\ddot x=-\frac{x}{r^3},\quad \ddot y=-\frac{y}{r^3},\quad\hbox{where }r=\sqrt{x^2+y^2},
\end{equation}
restricted to a bounded and closed interval of time. The {\it body} $q=(x,y)$ starts from an {\it initial point} $\A\in\O xy$ at a time $t_\A$ and arrives at a {\it final point} $\B\in \O xy$ at a time $t_\B>t_\A$.
\end{definition}

We restrict our study to a plane $\O xy$. This convention only makes a small difference in the question of the number of solutions of the Lambert problem: If $\A$, $\O$ and $\B$ are collinear in this order, any individual arc going from $\A$ to $\B$ in the plane $\O xy$ generates an infinite family of arcs in the space $\O xyz$, by merely rotating it around the line $\A\O\B$.

\begin{definition} A rectilinear arc is a Keplerian arc where the body remains on the same ray drawn from the center $\O$.
\end{definition}

Clearly, if $\A$ and $\B$ are distinct and on a same ray, all the arcs going from $\A$ to $\B$ are rectilinear.

\begin{definition} A Keplerian arc is called symmetric if its ends $\A$ and $\B$ are symmetric with respect to the principal axis of the Keplerian conic section.
\end{definition}

If $\A$ and $\B$ are distinct and at the same distance from $\O$, all the Keplerian arcs going from $\A$ to $\B$ are symmetric. This can be shown by using the focus-directrix property of a conic section: the two points being at the same distance from the focus $\O$, they are also at the same distance from the directrix.

\begin{definition} A Keplerian arc around $\O$ is called {\it indirect} if its convex hull contains $\O$. It is called {\it direct} otherwise. An indirect arc is called {\it multi-revolution} if the orbit is periodic and the elapsed time $t_\B-t_\A$ is greater than a period. It is called {\it simple} otherwise.
\end{definition}

These definitions apply to rectilinear arcs that are {\it extended after collision with the origin $\O$} (see \cite{alb}). The indirect rectilinear arcs are those which do collide.
For a nonrectilinear arc, we have the following characterization: the arc is direct if and only if the (positive) angle described along the arc is less than $\pi$.

\begin{proposition}\label{TDS}
In the plane $\O xy$, the direct Keplerian arcs around $\O$ whose ends $\A$ and $\B$ are distinct, symmetric with respect to the vertical axis $\O y$ and placed at a positive ordinate, are parametrized by the ``signed eccentricity'' $\eta$ varying in the interval $]-\infty,1[$. The elapsed time $T_D^S=t_\B-t_\A$ is an increasing function of $\eta$ with nonzero derivative. We have: $T_D^S\to 0$ when $\eta\to -\infty$ and $T_D^S\to +\infty$ when $\eta\to 1$.
\end{proposition}

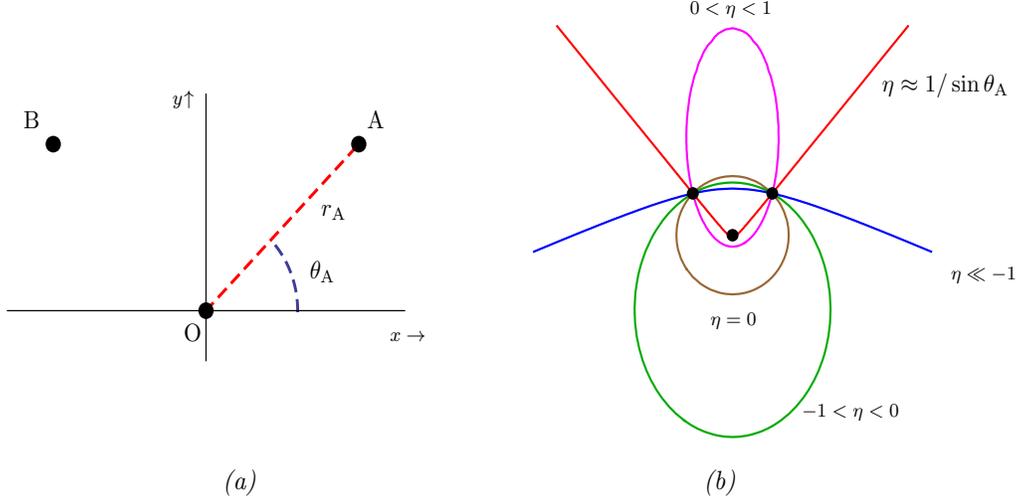
\begin{figure}
	\resizebox{5.5cm}{6cm}{\includegraphics{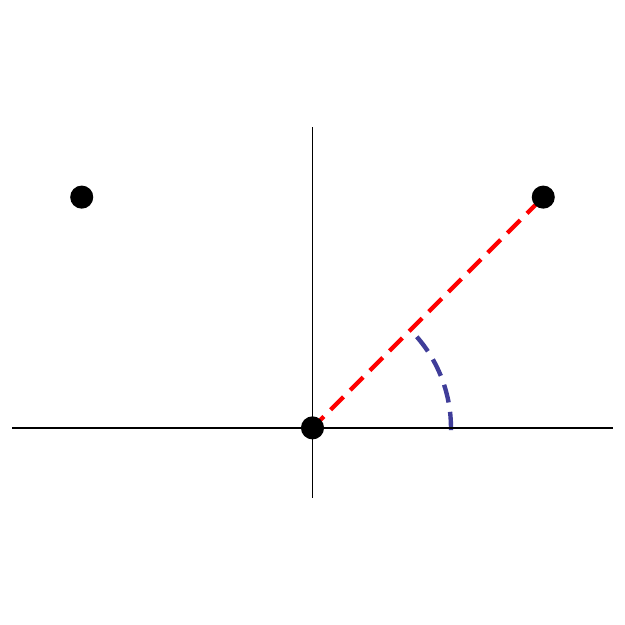}\put(-40,95){\footnotesize $r_\A$}\put(-45,70){\footnotesize $\theta_\A$}\put(-20,130){\footnotesize $\A$}\put(-170,130){\footnotesize $\B$}\put(-105,139.5){\scriptsize $y\hspace{-0.6mm}\uparrow$}\put(-10,45){\scriptsize $x\rightarrow$}\put(-100,45){\footnotesize $\O$}}\hspace{1.5cm}\resizebox{5.5cm}{6cm}{\includegraphics{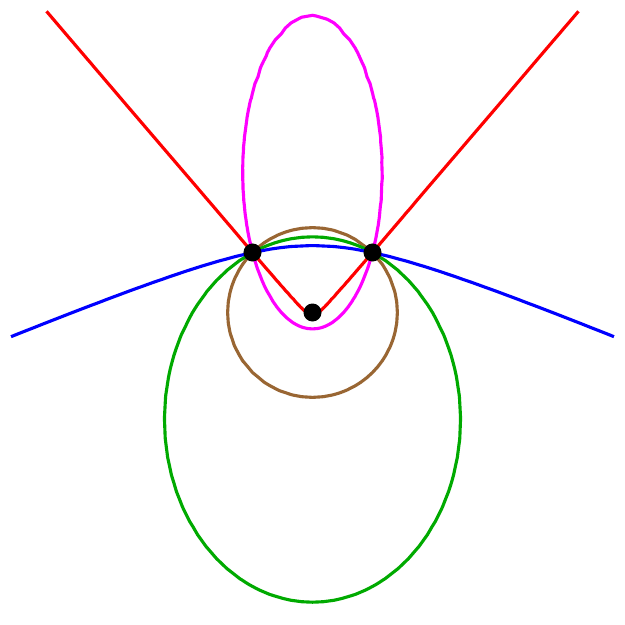}\put(5,72){\scriptsize $\eta\ll-1$}\put(-109,182){\scriptsize $0<\eta<1$}\put(-60,15){\scriptsize $-1<\eta<0$}\put(-100,53){\scriptsize $\eta=0$}\put(-25,150){\footnotesize $\eta\approx 1/\sin\theta_\A$}\put(-103,-15){\small\em (b)}\put(-313,-15){\small\em (a)}}
	\caption{{\em (a):} After a planar isometry, the endpoints $\A$, $\B$ of a symmetric Keplerian arc are located symmetrically with respect to the vertical axis. {\em (b):} The set of conic branches with a focus at $\O$ and passing through $\A$ and $\B$ can be parameterized by the signed eccentricity $\eta<1/\sin\theta_A$.}
\end{figure}

\begin{corollary} In the Euclidean plane or space, consider three points $\O$, $\A$, $\B$ forming a non-flat triangle with $\|\O\A\|=\|\O\B\|$.
There is a unique direct Keplerian arc around $\O$ going from $\A$ to $\B$ in a given positive elapsed time. This arc is in the plane $\O\A\B$ and is symmetric with respect to the perpendicular bisector of $\A \B$.
\end{corollary}

\begin{proof}
We already proved that the arc is symmetric with respect to the $\O y$ axis. It thus belongs to a conic section with polar equation
\begin{equation}
r=\frac{C^2}{1-\eta \sin \theta},
\end{equation}
with $C>0$ and $\eta\in \R$. The absolute value $|\eta |$ is the eccentricity, $C$ is the angular momentum, $C^2$ is the semi-parameter. The conic section passes through $\A$, of polar coordinates $(r_\A,\theta_\A)$, with $\theta_\A\in\,]0,\pi/2[$, and through $\B$, of polar coordinates $(r_\B,\theta_\B)=(r_\A,\pi -\theta_\A)$. This is expressed by the single condition
\begin{equation}
r_\A=\frac{C^2}{1-\eta \sin \theta_\A},
\end{equation}
which gives  $C^2$ as a function of $\eta\in\,]-\infty, 1/\sin\theta_A[$, and the new form of the polar equation
\begin{equation}
r=r_\A\frac{1-\eta \sin\theta_\A}{1-\eta\sin\theta}.
\end{equation}
At the ends of the interval, the arc is a limit of hyperbolas. When $\eta\to -\infty$, it is a segment going from $\A$ to $\B$. When $\eta\to 1/\sin\theta_\A$, it is a pair of segments, from $\A$ to $\O$ and then from $\O$ to $\B$ (see figure 1$b$). The direct arc is the upper arc, which exists if and only if $\eta\in\,]-\infty,1[$. We express the elapsed time along the direct symmetric arc by using the expression $C=r^2\dot\theta$ of the angular momentum:
\begin{equation}
T_D^S(\eta)=\int_{t_\A}^{t_\B} \d t=\int_{\theta_\A}^{\theta_\B}\frac{r^2}{C}\d\theta=\int_{\theta_\A}^{\theta_\B}\frac{r_\A^{3/2}(1-\eta \sin\theta_\A)^{3/2}}{(1-\eta\sin\theta)^2}\d\theta.
\end{equation}
We estimate the derivative $\d T^S_D/\d \eta$ by differentiating under the integration symbol. We introduce the shorter notation $p=\sin\theta_\A$, $q=\sin\theta$ and compute
$$\frac{\d}{\d \eta}\Bigl((1-\eta p)^{3/2}(1-\eta q)^{-2}\Bigr)=(1-\eta p)^{1/2}(1-\eta q)^{-3}K$$
where $$K=-\frac{3}{2}p(1-q\eta)+2q(1-p\eta)=\frac{1}{2}p(1-q\eta)+2(q-p).$$
On the considered arc, $0<p\leq q\leq 1$. Thus $K>0$ and consequently $\d T^S_D/\d \eta>0$.
\end{proof}

\section{Uniqueness of the direct rectilinear arc}
\label{UDRA}

We place the end-points $\A$ and $\B$ on the $\O x$ axis in this order: $0< x_\B<x_\A$. As already said, the direct arcs in the rectilinear case are simply the arcs without collision at $\O$. Some of these arcs are {\it culminating}, i.e., the body reaches somewhere a maximum distance from $\O$ with a zero velocity.  For a direct arc from $\A$ to $\B$, the velocity at $\A$, denoted by $v_\A$, is nonnegative if and only if the arc has a culmination point. There is an upper bound for $v_\A$, the {\it escape velocity} $v_E=\sqrt{2/x_\A}$.

\begin{proposition}\label{TDR} On the line $\O x$, the direct Keplerian arcs around $\O$, with ends $\A$ and $\B$ satisfying $0<x_\B <x_\A$, are parametrized by the initial velocity $v_\A\in\,]-\infty, v_E[$. The elapsed time $T_D^R=t_\B-t_\A$ is an increasing function of $v_\A$, with nonzero derivative. We have: $T_D^R\to 0$ when $v_\A\to -\infty$ and $T_D^R\to +\infty$ when $v_\A\to v_E$.
\end{proposition}

\begin{corollary}\label{UR} In the Euclidean plane or space, consider three distinct points $\O$,  $\A$, $\B$ such that $\A$ and $\B$ are on a same ray from $\O$. There is a unique direct Keplerian arc around $\O$ going from $\A$ to $\B$ in a given positive elapsed time. This arc is rectilinear. \end{corollary}

\begin{corollary} If two test particles attracted by a Newtonian fixed center start from the same point with two radial velocities, which we assume to be distinct, and if their motion is not extended after a collision with the center, they will not meet again.
\end{corollary}

\begin{proof} The velocity $v=\dot x$ is decreasing along the arc. We use it as a parameter:
\begin{equation}\label{velo}
T_D^R(v_\A)=\int_{t_\A}^{t_\B}\d t=\int_{v_\A}^{v_\B}\frac{\d t}{\d v}\d v=-\int_{v_\A}^{v_\B}\bigl(x(v)\bigr)^2\d v,
\end{equation}
according to the equation of motion $\ddot x=-1/x^2$. We introduce the variable $u=v-v_\A$.
\begin{equation}\label{velo1}
T_D^R(v_\A)=-\int_{0}^{v_\B-v_\A}\bigl(x(v_\A+u)\bigr)^2\d u.
\end{equation}
The conservation of energy $v_\A^2/2-1/x_\A=v_\B^2/2-1/x_\B$ gives
$v_\A\d v_\A=v_\B \d v_\B$. Similarly, $v_\A^2/2-1/x_\A=(v_\A+u)^2/2-1/x$
gives $-1/x_\A=v_\A u+u^2/2-1/x$ and, if $u$ is fixed, $u\d v_\A=-\d x/x^2$. So,
$$\frac{\d T_D^R}{\d v_\A}=-x_\B^2\frac{\d (v_\B-v_\A)}{\d v_\A}-2\int_{0}^{v_\B-v_\A}x\frac{\d x}{\d v_\A} \d u=-x_\B^2\Bigl(\frac{v_\A}{v_\B}-1\Bigr)+2\int_{0}^{v_\B-v_\A}x^3 u \d u.$$
Both terms are positive since $v_\A-v_\B>0$, $v_\B<0$ and $u\leq0$.
\end{proof}

\section{A first consequence of Lambert's theorem}

The proofs of propositions \ref{TDS} and \ref{TDR} are redundant. According to the following consequence of Lambert's theorem, these propositions can be deduced one from each other.

\begin{proposition}\label{lam1}
If the isosceles triangle $\A \O \B$ in proposition \ref{TDS} and the flat triangle $\O\B \A $ in proposition \ref{TDR} have same $\|\A \B\|$ and same $\|\O \A\|+\|\O \B\|$, then there is an invertible change of variable $\eta \leftrightarrow v_\A$ such that $T_D^S(\eta)=T_D^R(v_\A)$ and $H^S(\eta)=H^R(v_\A)$, where $T_D^S$, $H^S$, $T_D^R$, $H^R$ are respectively the elapsed time and the energy on the direct symmetric arc and on the direct rectilinear arc.
\end{proposition}

{\bf Explicit formulas.} We will get $T_D^S(\eta)=T_D^R(v_\A)$ and $H^S(\eta)=H^R(v_\A)$ with 
\begin{equation}\label{v}
v_\A(\eta)=\frac{\eta-\sqrt{x_\B/x_\A}}{\sqrt{{(x_\A+x_\B)/2}-\eta \sqrt{x_\A x_\B}}}.
\end{equation}
Here $x_\A$ and $x_\B$ are the positions of $\A$ and $\B$ in the rectilinear case (proposition \ref{TDR}). In the symmetric case (proposition \ref{TDS}), we used the polar coordinates $(r_\A,\theta_\A)$ of the initial point $\A$. We have
\begin{equation}\label{S}
\sqrt{\frac{x_\B}{x_\A}}=\tan\frac{\theta_\A}{2},\qquad \frac{x_\A+x_\B}{2}=r_\A,\qquad \sqrt{x_\A x_\B}=r_\A\sin \theta_\A,
\end{equation}
since by hypothesis $x_\A+x_\B=\|\O\A\|+\|\O\B\|=2r_\A$ and $x_\A-x_\B=\|\A\B\|=2r_\A\cos\theta_\A$. The expressions of the energy in both cases are respectively
\begin{equation}\label{H}
H^S(\eta)=\frac{\eta^2-1}{2C^2}=\frac{\eta^2-1}{2r_\A(1-\eta\sin\theta_\A)},
\qquad H^R(v_\A)=\frac{1}{2}v_\A^2-\frac{1}{x_\A}.
\end{equation}

\begin{proof} Lambert's theorem establishes the equality of the elapsed times for arcs with same $\|\A\B\|$, same $\|\O\A\|+\|\O\B\|$ and same energy. Actually, this statement is ambiguous: two arcs may have same energy and different elapsed times if they have different types. One should also specify that the corresponding arcs have same type, direct or indirect, simple or $n$-revolution, and also, in the elliptic case, same type with respect to the second focus (indirect if the convex hull of the arc contains the second focus, direct if it does not, see \cite{alb}, sect.\ 7). After using (\ref{S}) to substitute in (\ref{H}) we write
$$\frac{\eta^2-1}{x_\A+x_\B-2\eta \sqrt{x_\A x_\B}}=\frac{1}{2}v_\A^2-\frac{1}{x_\A}.$$
Thus $$\frac{1}{2}v_\A^2=\frac{x_\A(\eta^2-1)+x_\A+x_\B-2\eta\sqrt{x_\A x_\B}}{x_\A(x_\A+x_\B-2\eta \sqrt{x_\A x_\B})}=\frac{(\eta\sqrt{x_\A}-\sqrt{x_\B})^2}{x_\A(x_\A+x_\B-2\eta \sqrt{x_\A x_\B})}$$
This gives (\ref{v}) after a choice of sign which is such that the types of arcs coincide. The direct hyperbolic arcs correspond to $v_\A<-v_E$ and $\eta<-1$, the direct parabolic arc correspond to $v_\A=-v_E$ and $\eta=-1$, the direct elliptic arcs, direct with respect to the second focus, to $\eta\in\,]-1,\tan(\theta_\A/2)[$ and $v_\A\in\,]-v_E,0[$ and finally the direct elliptic arcs, indirect with respect to the second focus, to $\eta\in [\tan(\theta_\A/2),1[$ and to $v_\A\in [0,v_E[$.
\end{proof}

\section{Convexity results for the direct arc}
\label{convd}

\begin{proposition}\label{CTDR} The function $T^R_D(v_\A)$ of proposition \ref{TDR} is a convex function of $v_\A$ with non-zero second derivative.
\end{proposition}

\begin{proof} We continue the computation proving proposition \ref{TDR}. $$\frac{\d^2 T_D^R}{\d v_\A^2}=-\frac{x_\B^2}{v_\B^3}\bigl(v_\B^2-v_\A^2\bigr)+2x_\B^3(v_\B-v_\A)\Bigl(\frac{v_\A-v_\B}{v_\B}\Bigr)-6\int_{0}^{v_\B-v_\A}x^4u^2\d u.$$
The three terms are positive: $v_\B^2-v_\A^2=2/x_\B-2/x_\A>0$, $v_\B<0$, etc.
\end{proof}

We will give two proofs of the following proposition, which is analogous to proposition \ref{CTDR}.

\begin{proposition}\label{CTDS} The function $T_D^S(\eta)$ of proposition \ref{TDS} is a convex function of $\eta$ with non-zero second derivative.
\end{proposition}

\begin{proof}[Proof 1] With the same notation as in the proof of proposition \ref{TDS}, $$\frac{\d^2}{\d \eta^2}\Bigl((1-\eta p)^{3/2}(1-\eta q)^{-2}\Bigr)=$$
$$\frac{3}{4}(1-\eta p)^{-1/2}(1-\eta q)^{-2}p^2+6(1-\eta p)^{1/2}(1-\eta q)^{-4}q(q-p)>0.$$
\end{proof}

\begin{proof}[Proof 2] We use the convexity in proposition \ref{CTDR} and the convexity of the change of variable (\ref{v}), obtained by a straightforward computation (see figure 2$a$):
$$\frac{\d^2T^S_D}{\d \eta^2}=\frac{\d^2T^R_D}{\d v_\A^2}\Bigl(\frac{\d v_\A}{\d \eta}\Bigr)^2+\frac{\d T^R_D}{\d v_\A}\frac{\d^2 v_\A}{\d \eta^2}>0.$$
\end{proof}

\section{The simple indirect rectilinear arc}

The indirect rectilinear arcs are the Keplerian arcs where the body collides with the center $\O$. The motion is extended after collision. The body ``bounces'' with infinite velocity, while keeping the same total energy. The simple arcs are those which do not return to the same point with the same velocity. We have exactly the same uniqueness statement as in the direct case. We recall that $v_E=\sqrt{2/x_\A}$.

\begin{proposition}\label{TIR} On the line $\O x$, the simple indirect Keplerian arcs around $\O$, with ends $\A$ and $\B$ satisfying $0<x_\B <x_\A$, are parametrized by the initial velocity $v_\A\in\,]-\infty, v_E[$. The elapsed time $T_I^R=t_\B-t_\A$ is an increasing function of $v_\A$, with nonzero derivative. We have: $T_I^R\to 0$ when $v_\A\to -\infty$ and $T_I^R\to +\infty$ when $v_\A\to v_E$.
\end{proposition}

\begin{corollary}\label{UIR}
In the Euclidean plane or space, consider three distinct points $\O$,  $\A$, $\B$ such that $\A$ and $\B$ are on a same ray from $\O$. There is a unique simple indirect Keplerian arc around $\O$ going from $\A$ to $\B$ in a given positive elapsed time. This arc is rectilinear.
\end{corollary}

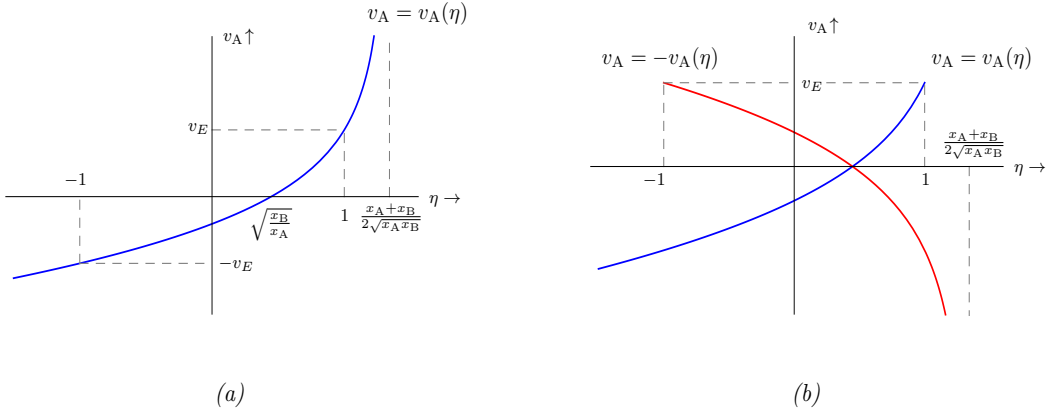
\begin{figure}
\resizebox{5.5cm}{6cm}{\includegraphics{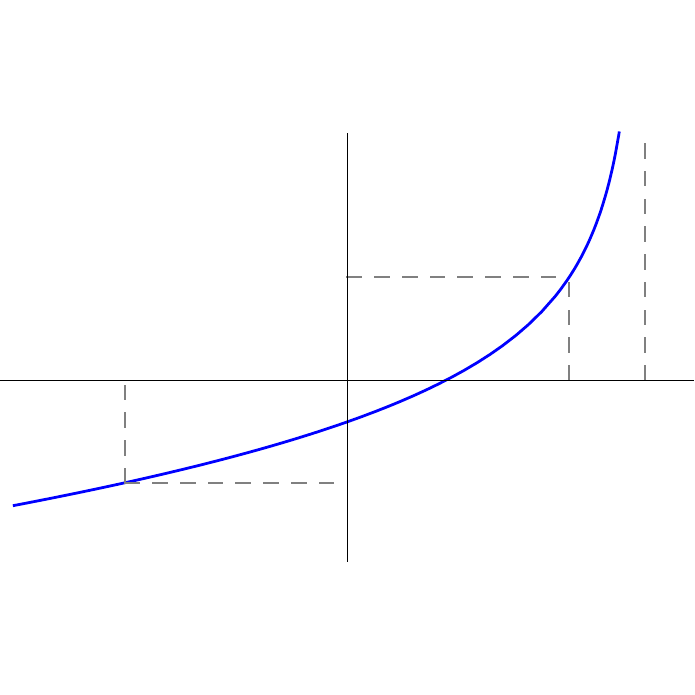}\put(-29,80){\scriptsize $\frac{x_\A+x_\B}{2\sqrt{x_\A x_\B}}$}\put(-38,80){\scriptsize $1$}\put(-82,77){\scriptsize $\sqrt{\frac{x_\B}{x_\A}}$}\put(-172,96){\scriptsize $-1$}\put(5,88){\scriptsize $\eta\rightarrow$}\put(-95,160){\scriptsize $v_\A\hspace{-0.6mm}\uparrow$}\put(-25,170){\footnotesize $v_\A=v_\A(\eta)$}\put(-112,120){\scriptsize $v_E$}\put(-97,59){\scriptsize $-v_E$}}
	\hspace{2cm}	\resizebox{5.5cm}{6cm}{\includegraphics{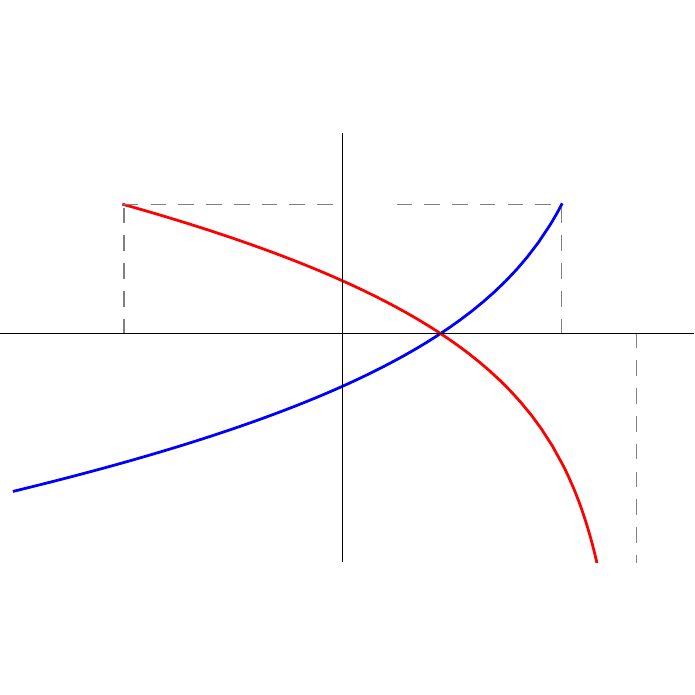}\put(-30,113){\scriptsize $\frac{x_\A+x_\B}{2\sqrt{x_\A x_\B}}$}\put(-40,96){\scriptsize $1$}\put(-175,96){\scriptsize $-1$}\put(5,102){\scriptsize $\eta\rightarrow$}\put(-95,165){\scriptsize $v_\A\hspace{-0.6mm}\uparrow$}\put(-35,150){\footnotesize $v_\A=v_\A(\eta)$}\put(-195,150){\footnotesize $v_\A=-v_\A(\eta)$}\put(-98,139.5){\scriptsize $v_E$}\put(-103,0){\small\em (b)}\put(-383,0){\small\em (a)}}
	\caption{{\em (a):} The graph of the function $v_\A=v_\A(\eta)$, $-\infty<\eta<\frac{x_\A+x_\B}{2\sqrt{x_\A x_\B}}$. {\em (b):} The initial speed as a function of $\eta$ for the direct arc ($v_\A=v_\A(\eta)$, $-\infty<\eta<1$), and the indirect arc ($v_\A=-v_\A(\eta)$, $-1<\eta<\frac{x_\A+x_\B}{2\sqrt{x_\A x_\B}}$).}
\end{figure}

\begin{proof} The non-culminating arcs start from $\A$ with $v_\A<0$, bounce off of $\O$, and arrive at $\B$ with a $v_\B>0$. As $2H=v^2-2x^{-1}$, the elapsed time is
\begin{equation}\label{TI}
T_I^R(v_\A)=\int_0^{x_\A} {dx\over \sqrt {2H+2x^{-1}}}+\int_0^{x_\B} {dx\over \sqrt {2H+2x^{-1}}}.
\end{equation}
If $v_\A<0$, $\d T_I^R/\d v_\A>0$ since $\d T_I^R/\d H<0$, and $\d H/\d v_\A<0$. If $v_\A\in[0,v_E[$, (\ref{TI}) is incorrect since the body is culminating. But the orbit is periodic with period ${\cal T}(v_\A)=2\pi(-2H)^{-3/2}$. We have
\begin{equation}\label{T}
T_I^R(v_\A)={\cal T}(v_\A)-T_D^R(-v_\A).
\end{equation}
The first term has a non-negative derivative with respect to $v_\A$, the second a positive derivative, according to proposition \ref{TDR}.
\end{proof}

{\bf Remark.} The argument proving proposition \ref{TDS} does not pass to the indirect case since the contributions to the integral near $\theta_\A$ and near $\theta_\B$ are increasing functions of $\eta$, while the contribution to the integral near $-\pi/2$ is a decreasing function of $\eta$.

\section{Uniqueness in the general case}
\label{UGC}

Here we use Lambert's theorem to deduce statements about the general arcs (neither rectilinear nor symmetric). Recall that a Keplerian arc is called simple if it does not pass twice at the same point with the same velocity. In the next statement the motion of the body is extended after collision.

\begin{theorem}
In the Euclidean plane, consider three distinct points $\O$, $\A$, $\B$. There are exactly two simple Keplerian arcs around $\O$ going from $\A$ to $\B$ in a given positive elapsed time.
\end{theorem}

\begin{proof} If $\O$ is not on the segment $\A\B$, there is exactly one direct arc and exactly one simple indirect arc, according to theorem \ref{DI} below. In the exceptional case $\O\in\,]\B,\A[$, a reflection sends a solution of the Lambert problem in the plane onto another. The argument proving theorem \ref{DI} extends, requiring only the easy extension of proposition \ref{TDR} to the case $\O=\B$. There are exactly two arcs, which are indirect.
\end{proof}

\begin{theorem}\label{DI}
 In the Euclidean plane or space, consider three distinct points $\O$, $\A$, $\B$ such that $\O$ is not on the segment $\A\B$.
There is a unique direct Keplerian arc around $\O$ and a unique simple indirect Keplerian arc around $\O$ going from $\A$ to $\B$ in a given positive elapsed time.
\end{theorem}

\begin{proof}
By Lambert's Theorem (see Theorems 1 and 2 and Sect.\ 7 in \cite{alb}), any arc in the general case may be continuously changed into a rectilinear arc of same $\|\A\B\|$, same $\|\O\A\|+\|\O\B\|$, same energy, same elapsed time and same type. This continuous change defines a bijection between the direct arcs, and a bijection between the simple indirect arcs. Thus, the existence and uniqueness corollaries \ref{UR} and \ref{UIR} in the rectilinear case give the result. 
\end{proof}

\section{Counter-examples to convexity}
\label{CEX}

\begin{lemma}\label{infini}
The second derivative with respect to $v_\A$, at $v_\A=0$, of the function $T_D^R(v_\A)$ defined in proposition \ref{TDR} tends to $+\infty$ when $\B\to \A$.
\end{lemma}

\begin{proof}
In proposition \ref{CTDR} an expression of $\d^2 T_D^R/\d v_\A^2$ was deduced by differentiating twice expression (\ref{velo1}). If instead we differentiate twice (\ref{velo}), and then make $v_\A=0$, we find
$$\frac{\d^2 T_D^R}{\d v_\A^2}\Big |_{v_\A=0}=-\frac{x_\B^2}{v_\B}+2\int_{v_\B}^{0}x^3\d v.$$
When $\B\to\A$, $v_\B\to 0$ negatively, the first term tends to $+\infty$ and the second term goes to zero.
\end{proof}

\begin{proposition}\label{nonc}
The function $T_I^R(v_\A)$ of proposition \ref{TIR}, giving the elapsed time on a simple indirect rectilinear arc as a function of the initial velocity $v_\A$, is not convex if $\B$ is close enough to $\A$.
\end{proposition}

\begin{proof}
Equation (\ref{T}) is $T_I^R(v_\A)={\cal T}(v_\A)-T_D^R(-v_\A)$.  Differentiating twice, making $v_\A=0$ and $\B\to \A$, the first term remains finite while the second tends to $-\infty$ according to the lemma. Thus  $T_I^R(v_\A)$ is not convex at $v_\A=0$ if $\B$ is close enough to $\A$.
\end{proof}

{\bf Remark.} Lemma \ref{infini} may be explained as follows. Consider the limit
of the function $T_D^R(v_\A)$ as $\B\to \A$. When $v_\A>0$, the limit is the elapsed time needed to start from $\A$ with a velocity $v_\A$, culminate and come back to $\A$. Using for example (\ref{velo}) we see that $T_D^R(v_\A)\sim 2x_\A^2 v_\A$ for small positive $v_\A$. For $v_\A<0$, the limit is zero. Consequently, the pointwise limit of the function $T_D^R(v_\A)$ is not differentiable at $v_\A=0$.

{\bf Remark.} The variable $v_\A$ was introduced in section \ref{UDRA} as a natural parameter of the family of rectilinear arcs from an exterior point $\A$ to an interior point $\B$. Section 4 introduces the parameter $v_\A(\eta)$ of the family of symmetric arcs with given symmetric ends $\A$ and $\B$. This is not the velocity at $\A$, but the parameter corresponding to the rectilinear $v_\A$ through Lambert's theorem. Section \ref{UGC} implicitly uses $v_\A$ as a parameter of the family of arcs with arbitrary given ends, since it reduces the general case to the rectilinear case by using Lambert's theorem. {\it Proposition \ref{CTDR}, lemma \ref{infini}, proposition \ref{nonc} and the above remark apply as well to a family of arcs with arbitrary fixed ends}, that we parametrize by $v_\A$. Proposition \ref{CTDR} and lemma \ref{infini} concern the family of direct arcs, while proposition \ref{nonc} concerns the family of simple indirect arcs. The conditions $\B \to \A$ and $\B$ sufficiently close to $\A$ remain unchanged.

{\bf Remark.} The statements of the above lemma, proposition and remarks are essentially due to Lancaster and Blanchard \cite{lb} (see their Figures 2 and 4). But this only appears if we notice that their variable $x$ is, up to a constant factor, our variable $v_\A$. Indeed, they define $x=\cos(\alpha/2)$ where $\alpha$ is defined by their formulas $(9)$ and $(10)$, due to Lagrange \cite{lagrange4}, p.\ 563, which translate into formulas the reduction to the rectilinear case proposed by Lambert's theorem. After this reduction, $\alpha$ is simply the eccentric anomaly of the exterior end $\A$, in the Keplerian rectilinear motion from the origin $\O$. Thus, $x_\A=a(1-\cos \alpha)$, where $a$ is the semi-major axis of the flat ellipse. Now the energy gives $v_\A^2=2/x_\A-1/a=2/x_\A-(1-\cos\alpha)/x_\A$ or $x_\A v_\A^2=2\cos^2(\alpha/2)$. Recall that the escape velocity is $v_E=\sqrt{2/x_\A}$. Lancaster and Blanchard's variable $x$ satisfies $v_\A=x v_E$. Similarly, their variable $q$ satisfies $q^2 x_\A =x_\B$.

\begin{proposition}
The function $T_I^S(\eta)$, giving the elapsed time on a simple indirect symmetric arc as a function of the signed eccentricity $\eta$, is not convex if $\theta_\A$ is close enough to $\pi/2$.
\end{proposition}

\begin{proof}
There is an analogue of proposition \ref{lam1} for simple indirect arcs. We have $T_I^S(\eta)=T_I^R(v_\A)$ with
\begin{equation}\label{mv}
v_\A(\eta)=\frac{\sqrt{x_\B/x_\A}-\eta}{\sqrt{{(x_\A+x_\B)/2}-\eta \sqrt{x_\A x_\B}}}.
\end{equation}
Compared to the direct case, the sign is changed, in agreement with the last argument in the proof of proposition \ref{lam1}. For the indirect arc, we still have $v_\A\in\,]-\infty,v_E[$, but $\eta\in\,]-1,1/\sin\theta_\A[$, where $\sin\theta_\A=2\sqrt{x_\A x_\B}/(x_\A+x_\B)$ according to (\ref{S}). The corresponding successions of types of arc can be checked using figure 2$b$. Now $v_\A$ is a concave function of $\eta$ and
$$\frac{\d^2T^S_I}{\d \eta^2}=\frac{\d^2T^R_I}{\d v_\A^2}\Bigl(\frac{\d v_\A}{\d \eta}\Bigr)^2+\frac{\d T^R_I}{\d v_\A}\frac{\d^2 v_\A}{\d \eta^2}<0\qquad\hbox{if}\quad\frac{\d^2T^R_I}{\d v_\A^2}<0.$$
When $x_\B\to x_\A$ in the rectilinear arc, $\theta_\A\to\pi/2$ in the corresponding symmetric arc. We have ${\d^2T^R_I}/{\d v_\A^2}<0$ at $v_\A=0$ according to the previous proposition. Thus $T_I^S$ is not convex.
\end{proof}

We established in section \ref{convd} the convexity of the elapsed time along the direct arc, as a function of $v_\A$ or $\eta$. Such a result may be useful since the Newton method for searching a root converges if applied to a convex function from an initial guess where the function is positive. It would be useful to have a similar result for the indirect arc. We have just shown that neither $v_\A$ nor $\eta$ is the variable we need. Sim\'o (see \cite{simo}, p.\ 242) has a claim which implies that the elapsed time along the indirect arc is convex in the variable $(u_\B-u_\A)^2$, where $u$ is the eccentric anomaly, defined on the Keplerian ellipses. According to him the proof is long. We cannot provide a short proof. 

\section{Pairs of multi-revolution arcs with same type}

\begin{proposition}\label{pair}
 In the Euclidean plane or space, consider three distinct points $\O$, $\A$, $\B$ such that $\O$ is not on the segment $\A\B$. Choose a positive integer $n$, and consider the Keplerian arcs around $\O$ starting from $\A$, passing again $n$ times through $\A$ with the initial velocity, and then going from $\A$ to $\B$ along a direct arc. There is a positive time $T_{\min}$ such that no such arc is travelled in an elapsed time $T<T_{\min}$,  a unique such arc is travelled in an elapsed time $T_{\min}$ and exactly two such arcs are travelled in any given elapsed time $T>T_{\min}$.
\end{proposition}

\begin{proof}
As in the proof of theorem \ref{DI}, Lambert's theorem reduces the proposition to the rectilinear problem with $0<x_\B<x_\A$, where the motion is extended after collision. The multi-revolution arcs require a periodic orbit. They are parametrized by $v_\A\in\,]-v_E,v_E[$, where $v_E=\sqrt{2/x_\A}$ is the escape velocity. The time of travel is $T(v_\A)=n{\cal T}(v_\A)+T_D^R(v_\A)$ where ${\cal T}(v_\A)=2\pi(2/x_\A-v_\A^2)^{-3/2}$ is the period. It tends to $+\infty$ at both ends of the interval.
Both terms, and consequently $T$, are convex functions of $v_\A$ with non-zero second derivative, according to the formula for ${\cal T}(v_\A)$ and to proposition \ref{CTDR}.\end{proof}

For completeness, we would like to state a similar proposition where ``a direct arc'' is changed into ``an indirect arc''. But in the proof we would need the convexity of $T_I^R$, the time of travel along an indirect arc. We do not have the analogue of proposition \ref{CTDR}, and we even gave relevant counter-examples in section \ref{CEX}. The convexity could be proved in another variable than $v_\A$. In particular, the convexity announced by Sim\'o  (see section \ref{CEX}) gives immediately the analogue of proposition \ref{pair}. We were able to check numerically this convexity as well as the convexity of the period ${\cal T}$. But, no simple proof is known, and no proof has been published, as far as we know.

Eliasberg \cite{eli} enunciates all the above conclusions and some others, notably, about total number of solutions of the Lambert problem, including all the types, simple or multi-revolution. This requires a study of the various values of $T_{\min}$, as a function of $n$ and of the type of the remaining arc, direct or indirect. However, as already commented in the introduction, it seems that part of the argument of \cite{eli} is based on a figure drawn in a particular case.

\section{A parameter for the general arc}
\label{parameter}

The continuous change of arc given by Lambert's theorem maps a direct general arc to a direct rectilinear arc or, as well, to a direct symmetric arc. This process allows to associate uniquely to such a general arc the parameter $v_\A$ of the rectilinear arc and, as well, the parameter $\eta$ of the symmetric arc. These two parameters will furthermore correspond one to each other by formula (\ref{v}).

There is a remarkably simple expression of the parameter $\eta$ from the parametrization by the eccentricity vector of the orbits passing through two points $\A$ and $\B$. The parameter $v_\A$ is subsequently obtained through (\ref{v}).

The nonrectilinear Keplerian branches in the plane $\O xy$ satisfy an equation
\begin{equation}
r=\alpha x+\beta y+\gamma
\end{equation}
where $r=\sqrt{x^2+y^2}>0$, $(\alpha,\beta)\in\R^2$ is the eccentricity vector and $\gamma=C^2>0$ is the semi-parameter of the conic section, which is also the square of the angular momentum. The conditions of passing through the points $\A=(x_\A,y_\A)$
and $\B=(x_\B,y_\B)$ are affine conditions in $(\alpha,\beta,\gamma)\in\R^3$:
\begin{equation}\label{AB}
r_\A=\alpha x_\A+\beta y_\A+\gamma, \quad r_\B=\alpha x_\B+\beta y_\B+\gamma.
\end{equation}
If  $\A\neq\B$, both equations are independent. The family of Keplerian branches passing through $\A$ and $\B$ is thus parametrized by an interval of a line in $\R^3$, defined by $(\ref{AB})$ and the inequality $\gamma>0$. If $\A$ and $\B$ are on a same ray from $\O$,
then $(\ref{AB})$ implies $\gamma=0$, which is forbidden.

 By choosing an appropriate frame $\O xy$, we may assume that $y_\A=y_\B\geq 0$, and, in the case $y_\A=y_\B=0$, that $x_\B<0<x_\A$. The difference of both equations (\ref{AB}) gives
\begin{equation}\label{alpha}
\alpha=\frac{r_\A-r_\B}{x_\A-x_\B}.
\end{equation}
The abscissa $\alpha$ of the eccentricity vector is thus the same for all the orbits passing through $\A$ and $\B$  (as observed in \cite{bfs}). Furthermore,  the triangular inequality gives $-1< \alpha<1$. Once deduced (\ref{alpha}), $\gamma$ is given by (\ref{AB}) as an affine function of $\beta$. The condition $\gamma>0$ delimitates an unbounded interval for $\beta$. If $y_\A=y_\B=0$ and $x_\A x_\B<0$, this interval is $\R$; otherwise, this interval is bounded from above. A convenient explicit expression of $\gamma$ is obtained as follows. We first notice the identity:
\begin{equation}\label{alphab}
\alpha=\frac{r_\A-r_\B}{x_\A-x_\B}=\frac{r_\A^2-r_\B^2}{(x_\A-x_\B)(r_\A+r_\B)}=\frac{x_\A+x_\B}{r_\A+r_\B}
\end{equation}
since $r_\A^2=x_\A^2+y_\A^2$, $r_\B^2=x_\B^2+y_\B^2$ and $y_\A=y_\B$. We add both expressions (\ref{AB}) and substitute using (\ref{alphab}):
\begin{equation}\label{sum}
r_\A+r_\B=\alpha^2 (r_\A+r_\B)+2 \beta y_\A+2\gamma.
\end{equation}
The classical expression of the energy
\begin{equation}
H=\frac{\alpha^2+\beta^2-1}{2\gamma}
\end{equation}
gives a rational expression in $\beta$:
\begin{equation}
H=\frac{\alpha^2+\beta^2-1}{(1-\alpha^2)(r_\A+r_\B)-2y_\A\beta}
\end{equation}
The ordinate $\beta$ of the eccentricity vector appears as an excellent parameter for the family of orbits passing through $\A$ and $\B$. It is advertised in \cite{ava}. But, following \cite{ala}, we take a step forward and consider
\begin{equation}\label{universel}
\hat\beta={\beta\over \sqrt{1-\alpha^2}}.
\end{equation}
For $y_\A=y_\B\geq 0$, by the area of the triangle $\O\A\B$ expressed by Heron formula, 
\begin{equation}
2y_\A|x_\A-x_\B|=\sqrt{(r_\A+r_\B)^2-(x_\A-x_\B)^2}\sqrt{(x_\A-x_\B)^2-(r_\A-r_\B)^2},
\end{equation} and by (\ref{alpha}),
\begin{equation}\label{Huniversel}
H={\hat\beta^2-1\over r_\A+r_\B-\hat\beta \sqrt{(r_\A+r_\B)^2-(x_\A-x_\B)^2}}.
\end{equation}

\begin{proposition}
In the Euclidean plane, during a continuous change of a Keplerian arc around $\O$, with ends $\A$ and $\B$, leaving $\|\A\B\|$, $\|\O\A\|+\|\O\B\|$ and $H$ constant, the parameter $\hat\beta$ given by (\ref{universel}) is also constant.\end{proposition}

\begin{proof}
Recall that in (\ref{universel}), $(\alpha,\beta)$ is the eccentricity vector in a frame $\O xy$ where the chord $\A\B$ is horizontal. In formula (\ref{Huniversel}), $|x_\A-x_\B|=\|\A\B\|$, $r_\A=\|\O\A\|$, $r_\B=\|\O\B\|$. The coefficients of (\ref{Huniversel}) are constant during the continuous change, and $H$ is also constant. Consequently, $\hat\beta$ is constant.
\end{proof}

We announced in the beginning of this section a simple way to get the parameter $\eta$ of the symmetric arc obtained by continuous change from a general arc. Since for the symmetric arc $\hat\beta=\beta=\eta$, this parameter is simply $\hat \beta$. The contrast between the simplicity of expression (\ref{universel}) and the subtlety of the deduction of its main property is quite astonishing. 

\section{Concluding remarks}

Lambert's theorem reduces the Lambert problem for a general triangle $\O\A\B$ to the cases of special triangles. The latter can be the isosceles triangles, with symmetry axis passing through $\O$, or the flat triangles, with $\B$ on the segment $\O\A$. Our study of qualitative questions, such as the number of solutions of the problem, shows that the reduction to the flat triangle gives better and simpler results than the reduction to the isosceles triangle. For example, the convexity of the traveling time function $v_\A\mapsto T_D^R$ on the direct arc appears as a stronger property than the corresponding property on the symmetric arc, since proposition \ref{CTDS} is obtained from proposition \ref{CTDR} in the second proof.

The reduction to the flat triangle may indeed be seen as the limit of a continuous change of Keplerian arc which finally gives a rectilinear arc. Two difficulties appear, which are easily resolved. The first is that the flat Keplerian arc, in the limit, may have a collision with the center. This difficulty is resolved by the extension after collision. The second happens if we think of complete orbits passing through two points $\A$ and $\B$, rather than of arcs going from $\A$ to $\B$. The interpretation of the parameter $\hat\beta$ in section \ref{parameter} induces such consideration. We should ask what is a rectilinear orbit passing through two points. The answer, again, is quite simple. Among the four simple arcs going from $\A$ to $\B$ in the negative energy case, we should consider that an arc belongs to the same orbit as its complementary. Thus, we have indeed four arcs belonging to two ``orbits'' which look the same, but should be distinguished. In the non-negative energy case, two similar orbits should be distinguished in the same way. To see these ``orbits'',
{\it we should cut a rectilinear orbit into two halves by cutting it at $\O$ and at the culmination point, if any.} We should then distinguish two cases: either $\A$ and $\B$ are placed on the same half, either they are placed on different halves. These two cases are the two ``orbits''. The sign of $v_\B$, or of the parameter $q$ in \cite{lb}, distinguishes between these two options.

\medskip

{\it Acknowledgement.} We wish to thank Christian Marchal for his comments. The second author has been partially supported by Spanish MINECO Grant with FEDER funds MTM2017-82348-C2-1-P. He is indebted to the IMCCE (Observatoire de Paris), and in particular is grateful to Alain Chenciner for his hospitality.

\end{document}